\pdfoutput=1
\documentclass[11pt,reqno]{amsart}

\usepackage[utf8]{inputenc}
\usepackage{lmodern}
\usepackage[T1]{fontenc}
\usepackage{amsmath,amssymb,amsthm,mathrsfs}
\usepackage[margin=1.1in]{geometry}
\usepackage{graphicx}
\usepackage[hidelinks]{hyperref}

\theoremstyle{plain}
\newtheorem{theorem}{Theorem}[section]
\newtheorem{proposition}[theorem]{Proposition}
\newtheorem{lemma}[theorem]{Lemma}
\newtheorem{corollary}[theorem]{Corollary}

\theoremstyle{remark}
\newtheorem{remark}[theorem]{Remark}

\numberwithin{equation}{section}

\begin{document}

\title[Holes in planar parallel sets]{Holes in planar parallel sets:
an integrated Betti-number bound and its pointwise failure}

\author{Tristan Guillaume}
\address{CY Cergy Paris Universit\'e, Laboratoire Thema,
33 boulevard du port, F-95011 Cergy-Pontoise Cedex, France}
\email{tristan.guillaume@cyu.fr}


\begin{abstract}
Let \(A\) be a nonempty compact subset of the plane and let \(A^{(r)}\) be its
parallel set at distance \(r\). We prove that the number of holes of
\(A^{(r)}\) --- the number of bounded components of its complement, which is
its first Betti number --- satisfies
\(\int_{r_{0}}^{\infty}\beta_{1}(A^{(r)})\,dr \leq 4050\,(\operatorname{diam}A)^{4}r_{0}^{-3}\)
for every \(r_{0}>0\), the integrand vanishing for
\(r \geq \operatorname{diam}A/\sqrt{3}\). The proof rests on Fu's theorem that
the critical values of the distance function of a planar compact set form a
set of vanishing half-dimensional Hausdorff measure, on two lemmas of Rataj,
Spodarev and Meschenmoser, and on a square-root summability estimate for the
gaps of the critical-value set, of which we give a complete proof. We show by
an explicit family of curves --- two combs facing each other --- that no
analogous bound can hold at a fixed radius: a connected curve of bounded
length, diameter, oscillation count, parallel-set area and parallel-set
perimeter can have arbitrarily many holes at one radius, so the integrated
estimate cannot be replaced by a fixed-radius bound in terms of these coarse
geometric quantities. We also bound the hole count uniformly in the radius by the
number of components of local maxima of the distance function, and record a
bound on the boundary length of a parallel set by its area. The results
supply the deterministic input for limit theorems on the persistent homology
of the Wiener sausage.
\end{abstract}

\maketitle

\section{Introduction}\label{sec:intro}

Let \(A \subset \mathbb{R}^{2}\) be a nonempty compact set. For \(r \geq 0\)
its closed parallel set at distance \(r\) is
\[A^{(r)}: = \{ x \in \mathbb{R}^{2}:\ \operatorname{dist}(x,A) \leq r\},\]
the union of the closed discs of radius \(r\) centred at the points of \(A\).
In computational topology the same set is called the \(r\)-offset of \(A\) and
the family \((A^{(r)})_{r \geq 0}\) the offset filtration \cite{ref12}; we use
the older word. Parallel sets are classical objects of convex and integral
geometry --- Steiner's formula, Kneser's inequality \cite{Kne51} and Stach\'o's
study of the parallel volume function \cite{Sta76} concern their area --- and
the way their topology changes with \(r\) is governed by the critical points of
the distance function \(d_{A}\), in the sense of Grove \cite{Gr93} and Fu
\cite{Fu85}: for radii in an interval free of critical values the parallel
sets are homotopy equivalent, indeed isotopic, by the isotopy lemma of critical
point theory, which is the mechanism behind the weak feature size of Chazal
and Lieutier \cite{CL07} and the sampling theory of Chazal, Cohen-Steiner and
Lieutier \cite{CCL09}. In the plane, Fu \cite{Fu85} proved that the set
\(C(A)\) of critical values has vanishing one-half-dimensional Hausdorff
measure, so that the topology of \(A^{(r)}\) is locally constant in \(r\) off
a set of radii that is not merely null but half-dimensionally null. Rataj and
Zaj\'{\i}\v{c}ek \cite{RZ20} sharpened this and showed the sharpening optimal:
the critical values of \(A\) above any fixed level \(s > 0\) form a compact
null set whose gaps have summable square roots, and conversely every such set
bounded away from \(0\) is contained in the critical-value set of some compact
planar set.

This paper is about the number of holes of \(A^{(r)}\): the number
\(\beta_{1}(A^{(r)})\) of bounded connected components of
\(\mathbb{R}^{2} \smallsetminus A^{(r)}\), which by Alexander duality is the
first Betti number of \(A^{(r)}\) in \v{C}ech cohomology. Since \(A^{(r)}\) is
contractible for \(r > \operatorname{diam}A\), holes are a transient
phenomenon, and the question we ask is how many of them a single compact set
can produce, and at which radii. The answer depends sharply on whether one
asks at a fixed radius or across radii. Write \(D = \operatorname{diam}A\).

\begin{theorem}[integrated bound]\label{thm:A}
For every compact \(A \subseteq \mathbb{R}^{2}\) and every \(r_{0} > 0\),
\[\int_{r_{0}}^{\infty}\beta_{1}\left( A^{(r)} \right)dr\, \leq \, 4050\,\frac{D^{4}}{r_{0}^{3}},\]
the integrand vanishing for \(r \geq D/\sqrt{3}\).
\end{theorem}

\begin{theorem}[no pointwise bound]\label{thm:B}
For every \(\varepsilon \in (0,1/10\rbrack\) there is a connected rectifiable
curve \(\gamma_{\varepsilon} \subset \mathbb{R}^{2}\), contained in a fixed
disc of radius \(4\), of length at most \(12\), and admitting at most \(30\)
increments at scale \(1/2\) in the sense of Section~\ref{sec:combs}, whose
unit parallel set has at least \(1/(2\varepsilon)\) holes, while the area
and the boundary length of that parallel set stay bounded. Consequently no
bound on \(\beta_{1}(A^{(r)})\) at a prescribed radius \(r > 0\) can be a
locally bounded function of the diameter, length and increment count of \(A\)
and of the area and boundary length of \(A^{(r)}\).
\end{theorem}

\begin{theorem}[scale-uniform bound]\label{thm:C}
For every compact \(A \subseteq \mathbb{R}^{2}\) and every \(s > 0\),
\[\sup_{r \geq s}\ \beta_{1}\left( A^{(r)} \right)\, \leq \, m_{A}(s),\]
where \(m_{A}(s)\) is the number of connected components of the set
\(\mathcal{M}_{A}\) of local maxima of \(d_{A}\) whose height --- constant on
each component, Lemma~\ref{lem:maxima} --- exceeds \(s\).
\end{theorem}

The ingredients of the proof of Theorem~\ref{thm:A} are Fu's theorem; two lemmas
of Rataj, Spodarev and Meschenmoser \cite{RSM09}, which bound the hole count at
a radius \(r\) by a multiple of \((D + 2r)^{2}/(r\sqrt{r\varepsilon})\)
whenever \(\lbrack r,r + \varepsilon)\) is free of critical values, and which
estimate the mutual distances of critical points at different heights; and a
square-root summability estimate for the gaps of the critical-value set
(Lemma~\ref{lem:gaps}), which we prove in full, because the corresponding
statement in \cite[Corollary 4.1]{RSM09} is printed without proof and, as
Remark~\ref{rem:rsmcor} explains, needs its scope restricted and its constant
adjusted. The point of Theorem~\ref{thm:A} is the statement: the singularity
\(\varepsilon^{-1/2}\) that the fixed-radius bound exhibits as \(r\) approaches
a critical value from below is integrable across radii precisely because the
critical gaps are square-root summable, and that exponent is the exponent of
Fu's theorem. The square-root summability of the critical gaps is not new:
Rataj and Zaj\'{\i}\v{c}ek \cite{RZ20} proved it in a sharp form --- sharp in
the embedding sense that every compact null set bounded away from \(0\) with
summable square-root gaps lies inside the critical-value set of some compact
planar set --- and their \cite[Proposition 3.16]{RZ20} gives a stronger,
weighted dyadic estimate with a \(D^{2}\)-dependence and a very large
constant; Lemma~\ref{lem:gaps} is the
special case needed here with the constant made explicit through the segment
lemma of \cite{RSM09}, and Section~\ref{sec:notes} explains why, by the
optimality in \cite{RZ20}, the route through gap sums cannot be sharpened.
To the best of our knowledge the integrated bound on the hole count itself has
not appeared; the closest statement we know is \cite[Proposition 3.5]{RZ20},
an integrated bound on the length of the set of critical points above height
\(r\). The constant and the exponents in Theorem~\ref{thm:A} are not sharp,
and the sharp form --- which we expect to be an area bound --- is stated as an
open problem in Remark~\ref{rem:sharp}.

Theorem~\ref{thm:B} shows that integration over the radius is not an artefact
of the proof. The curve \(\gamma_{\varepsilon}\) consists of two combs facing
each other, with teeth of length \(\varepsilon^{2}/2\) at spacing
\(\varepsilon\), joined by a detour; at radius \(1\) the parallel sets of the
two combs overlap only in narrow islands around the teeth and leave about
\(0.88/\varepsilon\) cells between them (Figure~\ref{fig:two-combs}). Along the
family the diameter, length and increment count of the curve and the area
and boundary length of its unit parallel set stay bounded while the hole count
diverges (Corollary~\ref{cor:combs}), so no expression in these quantities
that is bounded on bounded sets --- in particular no polynomial or monotone
one --- can bound the hole count at a fixed radius. Theorem~\ref{thm:C}
identifies the deterministic quantity
that does control the hole count uniformly in the radius: not the geometry of
\(A\) but the number of components of local maxima of \(d_{A}\) --- centres of
empty discs of radius \(> s\) --- a quantity which, by Theorem~\ref{thm:B},
is not itself controlled by the geometry of \(A\).

\emph{Motivation.} The question arose in the study of the persistent homology
of the planar Wiener sausage \cite{G26}, the parallel set of the trace of a
planar Brownian motion. There, hole counts of parallel sets of random compact
sets are integrated against a weight on the radius axis, and the pathwise
bound of Theorem~\ref{thm:A} --- applied to the trace of the path over a random
time interval, with \(D\) controlled by exponential moments --- supplies all the
integrability that the limit theorems of \cite{G26} require, while
Theorem~\ref{thm:B} explains why the corresponding fixed-radius theory of
\cite[Section 8]{G26} must be probabilistic. Nothing probabilistic enters the
present paper, which is written to be read independently of \cite{G26}.
Finiteness of the expected number of holes of the planar Wiener sausage at a
fixed radius is due to Rataj, Spodarev and Meschenmoser
\cite[Proposition 4.1 and Corollary 4.2]{RSM09}, by the same critical-value
calculus, and the almost-sure regularity of any fixed radius for the Brownian
trace is from \cite[Theorem 4.1]{RSS09}; Honzl \cite{Ho14} studies the mean
number of complementary components of the sausage as a function of the radius;
for the components of the complement of the Brownian trace itself see
Mountford \cite{Mo89}, Le Gall \cite{LG90} and, for their shapes, Holden,
Nacu, Peres and Salisbury \cite{HNPS19}.

\emph{Organization.} Section~\ref{sec:conv} fixes conventions --- \v{C}ech
cohomology and Alexander duality --- and records the Mayer--Vietoris
inequalities we use. Section~\ref{sec:critical} assembles the critical-point
calculus: contact geometry, Fu's theorem, a lemma identifying singular and
\v{C}ech Betti numbers at noncritical radii (for readers who prefer singular
homology), the two lemmas of \cite{RSM09}, and the gap-summability lemma.
Section~\ref{sec:integrated} proves Theorem~\ref{thm:A} and records two
further consequences of the same calculus: star-shapedness of parallel sets
at radii above \(D/\sqrt{3}\), and a bound on the boundary length of a parallel
set by its area, from Kneser's inequality. Section~\ref{sec:maxima} proves
Theorem~\ref{thm:C}. Section~\ref{sec:combs} constructs the two-combs family
and proves Theorem~\ref{thm:B}, the uniformity of the coarse quantities being
Corollary~\ref{cor:combs}. Section~\ref{sec:notes} collects notes and
open questions.

Throughout, \(A,B\) denote nonempty compact subsets of \(\mathbb{R}^{2}\),
\(d_{A}(x): = \min\left\{ \lVert x - a \rVert :a \in A \right\}\) the distance
function, \(A^{(r)}: = \left\{ d_{A} \leq r \right\}\) the closed parallel set,
\(D = \operatorname{diam}A\), and \(\overline{B}(x,r)\) the closed disc of
radius \(r\) centred at \(x\).

\section{Conventions, and the Mayer--Vietoris toolkit}\label{sec:conv}

\begin{remark}[cohomological conventions]\label{rem:conv}
 Betti numbers
\(\beta_{i}\) are ranks of \v{C}ech (equivalently, Alexander--Spanier)
cohomology with coefficients in a fixed field. Two classical facts are
used repeatedly and record the reason for this choice. First, for
compact \(X \subseteq \mathbb{R}^{2} \subset S^{2}\), Alexander duality
gives
\({\check{H}}^{q}(X) \cong {\widetilde{H}}_{1 - q}\left( S^{2} \smallsetminus X \right)\);
taking \(q = 1\) identifies \(\beta_{1}(X)\) with the number of
bounded connected components of
\(\mathbb{R}^{2} \smallsetminus X\), and taking \(q = 2\) gives
\({\check{H}}^{2}(X) = 0\), since
\(S^{2} \smallsetminus X \neq \varnothing\). Second, the Mayer--Vietoris
sequence is exact for pairs of closed subsets of a paracompact
space in \v{C}ech cohomology (see \cite[Chapter 6]{Sp81}); no niceness of the
sets is assumed anywhere in this paper. When \(X\) is compact,
\(\beta_{0}(X) = \dim{\check{H}}^{0}(X)\) is the unreduced number of
connected components, possibly infinite; all inequalities below are read in
\(\lbrack 0,\infty\rbrack\). (Persistence theory usually
works with singular homology and, in degree zero, with reduced homology; the
two degree-zero conventions differ by one on a nonempty set.) In degree one,
where our main results live, Lemma~\ref{lem:bridge} shows that the singular
and \v{C}ech first Betti numbers of a parallel set agree at every noncritical
radius, so a reader who prefers singular homology may read every statement
below about \(\beta_{1}(A^{(r)})\) in that theory at almost every \(r\).
\end{remark}

\begin{lemma}[Mayer--Vietoris inequality]\label{lem:mv}
 Let
\(U,V \subseteq \mathbb{R}^{2}\) be compact with \(U\),
\(V\) and \(U \cup V\) connected. Then

\[\beta_{1}(U \cup V)\, \leq \,\beta_{1}(U) + \beta_{1}(V) + \beta_{0}(U \cap V) - 1.\]
\end{lemma}

\begin{proof}
If \(\beta_{0}(U \cap V) = \infty\) there is nothing to
prove; note \(U \cap V \neq \varnothing\) since \(U \cup V\) is
connected. In the exact sequence
\[{\check{H}}^{0}(U \cup V) \rightarrow {\check{H}}^{0}(U) \oplus {\check{H}}^{0}(V) \rightarrow {\check{H}}^{0}(U \cap V)\overset{\ \partial\ }{\rightarrow}{\check{H}}^{1}(U \cup V) \rightarrow {\check{H}}^{1}(U) \oplus {\check{H}}^{1}(V),\]
the first map has rank \(1\) and the middle space has dimension \(2\),
so the map into \({\check{H}}^{0}(U \cap V)\) has rank \(1\) and
\(\operatorname{rank}\partial \leq \beta_{0}(U \cap V) - 1\). Exactness at
\({\check{H}}^{1}(U \cup V)\) gives
\[\beta_{1}(U \cup V) = \operatorname{rank}\partial + \operatorname{rank}\left( {\check{H}}^{1}(U \cup V) \rightarrow {\check{H}}^{1}(U) \oplus {\check{H}}^{1}(V) \right) \leq \left( \beta_{0}(U \cap V) - 1 \right) + \beta_{1}(U) + \beta_{1}(V).\]
\end{proof}

\begin{lemma}[Mayer--Vietoris bound in degree one]\label{lem:signed}
~ Let
\(A,B \subset \mathbb{R}^{2}\) be compact, with Betti numbers read as ranks of
\v{C}ech cohomology as in Remark~\ref{rem:conv}, and assume
\begin{equation}\label{eq:fin}
\beta_{1}(A) < \infty\quad\text{and}\quad \beta_{1}(B) < \infty .
\end{equation}
Then
\[- \beta_{1}(A \cap B)\  \leq \ \beta_{1}(A \cup B) - \beta_{1}(A) - \beta_{1}(B)\  \leq \ \beta_{0}(A \cap B),\]
and in particular
\[\left| \beta_{1}(A \cup B) - \beta_{1}(A) - \beta_{1}(B) \right|\  \leq \ \beta_{1}(A \cap B) + \beta_{0}(A \cap B).\]
\end{lemma}

The finiteness hypothesis \eqref{eq:fin} is not cosmetic: the proof computes
a difference of ranks, and that arithmetic has no meaning when infinite
dimensions occur. For parallel sets it is no restriction in practice: by
Theorem~\ref{thm:main}, \(\beta_{1}(A^{(r)}) < \infty\) for Lebesgue-almost
every \(r > 0\), so \eqref{eq:fin} holds for the pair \(A^{(r)},B^{(r)}\) off
a null set of radii, which is all that an integral over the radius sees.

\begin{proof}
By Remark~\ref{rem:conv} the Mayer--Vietoris sequence
\[{\check{H}}^{1}(A \cap B)\overset{\ \alpha\ }{\leftarrow}{\check{H}}^{1}(A) \oplus {\check{H}}^{1}(B)\overset{\ \beta\ }{\leftarrow}{\check{H}}^{1}(A \cup B)\overset{\ \gamma\ }{\leftarrow}{\check{H}}^{0}(A \cap B)\]
is exact for the pair of closed sets \(A,B\); we argue with the corresponding
maps in the direction of increasing degree and write \(\beta_{i}\) for ranks.
Exactness gives \(\beta_{1}(A \cup B) = \dim\operatorname{im}\beta + \dim\operatorname{im}\gamma\)
and \(\dim\operatorname{im}\beta = \beta_{1}(A) + \beta_{1}(B) - \dim\operatorname{im}\alpha\),
whence
\(\beta_{1}(A \cup B) - \beta_{1}(A) - \beta_{1}(B) = \dim\operatorname{im}\gamma - \dim\operatorname{im}\alpha\).
Now \(\dim\operatorname{im}\alpha \leq \beta_{1}(A \cap B)\) and
\(\dim\operatorname{im}\gamma \leq \beta_{0}(A \cap B)\). \end{proof}

\begin{corollary}[interface bounds]\label{cor:interface}
 For all compact \(U,V \subseteq \mathbb{R}^{2}\),

\[\beta_{1}(U \cap V)\, \leq \,\beta_{1}(U) + \beta_{1}(V),\quad\quad\beta_{0}(U \cap V)\, \leq \,\beta_{0}(U) + \beta_{0}(V) + \beta_{1}(U \cup V).\]
\end{corollary}

\begin{proof}
The segment
\[{\check{H}}^{1}(U) \oplus {\check{H}}^{1}(V) \rightarrow {\check{H}}^{1}(U \cap V) \rightarrow {\check{H}}^{2}(U \cup V)\]
is exact and ends in \(0\) by Remark~\ref{rem:conv}, so the first map is
surjective; this is the first inequality. The segment
\[{\check{H}}^{0}(U) \oplus {\check{H}}^{0}(V) \rightarrow {\check{H}}^{0}(U \cap V) \rightarrow {\check{H}}^{1}(U \cup V)\]
gives the second.
\end{proof}

For compact \(A,B\) and \(r > 0\) the corollary applies to
\(U = A^{(r)}\), \(V = B^{(r)}\), for which \(U \cup V = (A \cup B)^{(r)}\):
the intersection of two parallel sets is not itself a parallel set, but its
Betti numbers are dominated by hole counts of honest parallel sets, to which
Theorem~\ref{thm:main} applies. When the two sets are connected the second
inequality improves by one.

\begin{lemma}[planar intersection complexity for connected sets]\label{lem:connected}
~
Let \(A,B \subset \mathbb{R}^{2}\) be compact connected sets. With Betti
numbers valued in \([0,\infty]\) as in Remark~\ref{rem:conv},
\(\beta_{0}(A \cap B) \leq \beta_{1}(A \cup B) + 1\) if
\(A \cap B \neq \varnothing\), and
\(\beta_{1}(A \cap B) \leq \beta_{1}(A) + \beta_{1}(B)\). Consequently

\[\beta_{1}(A \cap B) + \beta_{0}(A \cap B)\  \leq \ \beta_{1}(A) + \beta_{1}(B) + \beta_{1}(A \cup B) + 1.\]
\end{lemma}

\begin{proof}
Both assertions refine Corollary~\ref{cor:interface}: the segment
\[{\check{H}}^{0}(A) \oplus {\check{H}}^{0}(B) \rightarrow {\check{H}}^{0}(A \cap B) \rightarrow {\check{H}}^{1}(A \cup B)\]
of the Mayer--Vietoris sequence gives
\(\beta_{0}(A \cap B) \leq \beta_{0}(A) + \beta_{0}(B) + \beta_{1}(A \cup B) - 1 = \beta_{1}(A \cup B) + 1\),
the two sets being connected and meeting, so that, as in the proof of
Lemma~\ref{lem:mv}, the map into \(\check{H}^{0}(A \cap B)\) has rank one; and
the segment
\[{\check{H}}^{1}(A) \oplus {\check{H}}^{1}(B) \rightarrow {\check{H}}^{1}(A \cap B) \rightarrow {\check{H}}^{2}(A \cup B) = 0\]
gives the second. Summing the two yields the display. \end{proof}
\section{Critical points of the distance function}\label{sec:critical}

Following \cite{Fu85,RSM09}, for
\(x \in \mathbb{R}^{2} \smallsetminus A\) let
\[\Sigma_{A}(x): = \left\{ a \in A: \lVert x - a \rVert = d_{A}(x) \right\}\]
be the (compact, nonempty) set of nearest points. A point \(x\) is
critical if \(x \in \operatorname{conv}\Sigma_{A}(x)\), and regular
otherwise; \(r > 0\) is a critical value if some critical point
has \(d_{A}(x) = r\), and \(C(A) \subseteq (0,\infty)\) denotes the set
of critical values. (By \cite[Lemma 4.2]{Fu85}, this coincides with
criticality of the locally Lipschitz function \(d_{A}\) in the sense of
Clarke, but we will not need that.)

\begin{lemma}[contact geometry and height bound]\label{lem:contact}
 Let \(x\)
be a critical point with \(d: = d_{A}(x) > 0\), and let
\(V \subseteq S^{1}\) be the set of directions
\((p - x)/d\), \(p \in \Sigma_{A}(x)\). Then:

(i) every closed half-circle of \(S^{1}\) meets
\(V\); equivalently, every open arc disjoint from \(V\) has
length at most \(\pi\);

(ii) there exist \(p,q \in \Sigma_{A}(x)\) whose directions
are separated by an angle \(\geq 2\pi/3\); in particular
\(\lVert p - q \rVert \geq \sqrt{3}\, d\);

(iii) if \(A^{(r)}\) has a bounded complementary component then
\(C(A) \cap (r,\infty) \neq \varnothing\); consequently
\(\beta_{1}\left( A^{(r)} \right) = 0\) for every
\(r \geq \sup C(A)\) (with \(\sup\varnothing: = 0\)), and
\(\sup C(A) \leq D/\sqrt{3}\), so in particular
\(\beta_{1}\left( A^{(r)} \right) = 0\) for every \(r \geq D/\sqrt{3}\).
\end{lemma}

\begin{proof}
(i) \(V\) is closed. If \(V\) avoided the closed
half-circle \(\left\{ v:\langle v,u\rangle \geq 0 \right\}\), every
\(v \in V\) would have \(\langle v,u\rangle < 0\), hence so would every
convex combination; but criticality says \(0 \in \operatorname{conv}V\) ---
contradiction. An open arc of length \(> \pi\) contains a closed
half-circle in its interior, giving the restatement.

(ii) Suppose all pairs of \(V\) are separated by \(< 2\pi/3\). Fix
\(v_{0} \in V\) at angle \(0\); then \(V\) avoids the closed arc
\(\lbrack 2\pi/3,\, 4\pi/3\rbrack\), which lies in a maximal open
\(V\)-free arc \((\alpha,\beta)\) with, by (i),
\(\beta - \alpha \leq \pi\). Hence
\(\alpha \in \lbrack\pi/3,2\pi/3\rbrack\),
\(\beta \in \lbrack 4\pi/3,5\pi/3\rbrack\), and \(\alpha,\beta \in V\)
since \(V\) is closed; their separation is
\(\beta - \alpha \in \lbrack 2\pi/3,\pi\rbrack\), which as a geodesic
distance equals \(\beta - \alpha \geq 2\pi/3\) --- contradiction. The
chord is \(2d\sin(\pi/3) = \sqrt{3}\, d\).

(iii) Suppose \(A^{(r)}\) has a bounded complementary component \(L\). Then
\(L\) is open, \(\overline{L}\) is compact, \(d_{A} > r\) on \(L\) and
\(d_{A} = r\) on \(\partial L\). Let \(x \in \overline{L}\) maximize \(d_{A}\)
over \(\overline{L}\); then \(d_{A}(x) = r + \varepsilon\) for some
\(\varepsilon > 0\), so \(x \in L\)
is an interior point and hence a local maximum of \(d_{A}\) in
\(\mathbb{R}^{2}\). A local maximum is critical: if the contact directions
\(V\) lay in an open half-circle \(\{ v:\langle v,u\rangle > 0\}\), then for
small \(t > 0\) every point of \(A\) would be farther from \(x - tu\) than from
\(x\) --- the points of \(A\) near \(\Sigma_{A}(x)\) by the inequality
\(\lVert x - tu - a \rVert^{2} = \lVert x - a \rVert^{2} + 2t\langle a - x,u\rangle + t^{2}\)
and continuity of \(a \mapsto \langle a - x,u\rangle/\lVert a - x \rVert\),
the remaining points of \(A\) by compactness, being at distance
\(\geq d_{A}(x) + \eta\) for some \(\eta > 0\) --- so that
\(d_{A}(x - tu) > d_{A}(x)\), contradicting local maximality;
cf.~\cite[proof of Lemma 4.1]{RSM09}. Thus
\(r + \varepsilon \in C(A) \cap (r,\infty)\), which is the first assertion,
and its contrapositive gives \(\beta_{1}\left( A^{(r)} \right) = 0\) for
\(r \geq \sup C(A)\). Finally \(\sqrt{3}\, d \leq \lVert p - q \rVert \leq D\)
for every critical point, by (ii), so \(\sup C(A) \leq D/\sqrt{3}\).
\end{proof}

We next quote three imported results, in our notation. The first is
the deep one.

\begin{theorem}[Fu {\cite[Theorem 4.1]{Fu85}}]\label{thm:fu}
 For every
compact \(A \subseteq \mathbb{R}^{2}\), the set
\(C(A) \cup \{ 0\}\) is compact and
\(\mathbf{H}^{1/2}\left( C(A) \right) = 0\), where
\(\mathbf{H}^{1/2}\) is the one-half-dimensional Hausdorff
measure. In particular \(C(A)\) is Lebesgue-null.
\end{theorem}

The following lemma bridges the \v{C}ech convention of
Remark~\ref{rem:conv} and singular homology, in which persistence theory is
usually set up. It is not used in the proofs below; its role is explained in
Remark~\ref{rem:bridge}.

\begin{lemma}[singular and \v{C}ech Betti numbers agree at noncritical radii]\label{lem:bridge}
Let \(A \subset \mathbb{R}^{2}\) be nonempty and compact, and let
\(r \in (0,\infty) \smallsetminus C(A)\). Then \(A^{(r)}\) is a compact
absolute neighbourhood retract, and for every \(q \geq 0\) and every field
\(\mathbb{F}\),
\[\operatorname{rank}H_{q}\left( A^{(r)};\mathbb{F} \right)\  = \ \dim{\check{H}}^{q}\left( A^{(r)};\mathbb{F} \right).\]
In particular the two conventions assign the same value to
\(\beta_{1}\left( A^{(r)} \right)\), and by Alexander duality that common value
is the number of bounded connected components of
\(\mathbb{R}^{2} \smallsetminus A^{(r)}\). By Theorem~\ref{thm:fu} this holds
for Lebesgue-almost every \(r > 0\).
\end{lemma}

\begin{proof}
For \(r > 0\) the set \(A^{(r)} = \left\{ d_{A} \leq r \right\}\) is the closure
of the open set \(\left\{ d_{A} < r \right\}\): if \(d_{A}(x) = r\) and
\(p \in A\) is a nearest point, the points of the segment \(\lbrack p,x\rbrack\)
near \(x\) have \(d_{A} < r\). For \(r \notin C(A)\) one has
\(\partial A^{(r)} = \left\{ d_{A} = r \right\}\), since a point with
\(d_{A} = r\) interior to \(A^{(r)}\) would be a local maximum of \(d_{A}\),
hence critical; and the level set \(\left\{ d_{A} = r \right\}\) of the locally
Lipschitz function \(d_{A}\) at a noncritical value is a Lipschitz
\(1\)-manifold by Clarke's implicit function theorem, \cite[Theorem 3.1]{Fu85}
(see \cite[eq.~(21)]{RZ20} for this formulation). This refines Ferry's theorem
\cite{Fer76}, which gives the conclusion for almost every \(r\) in
\(\mathbb{R}^{2}\) and \(\mathbb{R}^{3}\); Ferry also exhibits a compact set in
\(\mathbb{R}^{4}\) for which it fails, which is one reason the present lemma is
stated in the plane. A compact planar set that is
the closure of its interior and has Lipschitz boundary is therefore a compact
topological \(2\)-manifold with boundary: near an interior point it is locally
Euclidean, and near a boundary point it is, in suitable coordinates, the region
below the graph of a Lipschitz function, which is homeomorphic to a half-disc.
Compact topological manifolds with boundary are absolute neighbourhood
retracts, hence locally contractible, and for a compact locally contractible
space \v{C}ech and singular theories agree; over a field, ranks of homology and
dimensions of cohomology then coincide by universal coefficients. The final
assertion is Alexander duality as recalled in Remark~\ref{rem:conv}.
\end{proof}

\begin{remark}[what the bridge is and is not used for]\label{rem:bridge}
Lemma~\ref{lem:bridge} is not needed anywhere in this paper, which is
written entirely in \v{C}ech terms, where Mayer--Vietoris for closed pairs and
Alexander duality are available for arbitrary compacta and no regularity is
assumed. Its role is to make the objects of persistence theory --- the
persistence module of the offset filtration, its diagram, and the Betti curve
\(r \mapsto \operatorname{rank}H_{1}(A^{(r)};\mathbb{F})\) --- the same
objects as those appearing in the estimates here, so that a radius integral of
the Betti curve, such as the smoothed functionals of \cite{G26}, may be read
either way; because such integrals are insensitive to null sets of radii,
agreement at almost every radius suffices. A reader who prefers to define the
persistence theory itself in \v{C}ech terms may dispense with the lemma
altogether.
\end{remark}

\begin{lemma}[Rataj--Spodarev--Meschenmoser {\cite[Lemma 4.1]{RSM09}}]\label{lem:rsm1}
 Let \(A \subseteq \mathbb{R}^{2}\) be compact and
let \(0 < \varepsilon < r\) satisfy
\(C(A) \cap \lbrack r,\, r + \varepsilon) = \varnothing\). Then
the number of bounded connected components of
\(\mathbb{R}^{2} \smallsetminus A^{(r)}\) is at most
\(\pi(D + 2r)^{2}/\left( r\sqrt{2r\varepsilon} \right)\).
\end{lemma}

\begin{lemma}[Rataj--Spodarev--Meschenmoser {\cite[Lemmas 4.2 and 4.3]{RSM09}}]\label{lem:rsm2}
 (i) If \(x,y\) are critical points of \(d_{A}\)
at heights \(0 < h < h'\), with
\(p \in \Sigma_{A}(x)\), \(q \in \Sigma_{A}(y)\), then the
segments \(S = \left\lbrack x,(x + p)/2 \right\rbrack\) and
\(T = \left\lbrack y,(y + q)/2 \right\rbrack\) satisfy
\(\operatorname{dist}(S,T) \geq \sqrt{h/2}\,\sqrt{h' - h}\).

(ii) Let \(S_{0},\ldots,S_{n}\) be segments of length
\(\geq s' > 0\) contained in a ball \(B_{R}\),
\(R \geq s'\), and
\(\varepsilon_{1},\ldots,\varepsilon_{n} > 0\) satisfy
\[\operatorname{dist}\left( S_{i},S_{j} \right) \geq \sqrt{s'/2}\,\sqrt{\varepsilon_{i + 1} + \cdots + \varepsilon_{j}}\]
for all \(0 \leq i < j \leq n\). Then
\(\sum_{i = 1}^{n}\sqrt{\varepsilon_{i}} \leq 8(s')^{- 3/2}\pi R^{2}\).
\end{lemma}

Both parts are proved in full in \cite{RSM09} (part (i) building on
Ferry's lemma \cite{Fer76}, reproved as \cite[Lemma 4.3]{Fu85}). From these we now
derive the gap-summability estimate in the form we need. Its content ---
that the gaps of \(C(A) \cap \lbrack s,\infty)\) have summable square roots
--- is a theorem of Rataj and Zaj\'{\i}\v{c}ek \cite[Theorem 1.1]{RZ20}, who
also show that this smallness is optimal; what the lemma adds is an explicit
constant, obtained from Lemma~\ref{lem:rsm2} rather than from the DC-function
methods of \cite{RZ20}.

\begin{lemma}[square-root summability of critical gaps]\label{lem:gaps}

Let \(A \subseteq \mathbb{R}^{2}\) be compact with
\(D > 0\), and \(0 < s \leq D\). Let
\(\left\{ I_{i} = \left( a_{i},b_{i} \right) \right\}_{i}\) be the
maximal open intervals of \((s,\infty) \smallsetminus C(A)\) that
are bounded with \(b_{i} \in C(A)\), and
\(\varepsilon_{i} = b_{i} - a_{i}\). Then

\[\sum_{i}^{}\sqrt{\varepsilon_{i}}\, \leq \, 32\pi\, s^{- 3/2}D^{2}.\]
\end{lemma}

\begin{proof}
By Theorem~\ref{thm:fu} the family is well defined and countable;
it suffices to bound finite subfamilies \(I_{1},\ldots,I_{m}\) ordered
by \(b_{1} < \cdots < b_{m}\). Each \(b_{i}\) is a critical value, so
pick a critical point \(x_{i}\) at height \(b_{i}\), a nearest point
\(p_{i} \in \Sigma_{A}\left( x_{i} \right)\), and the segment
\(S_{i} = \left\lbrack x_{i},\left( x_{i} + p_{i} \right)/2 \right\rbrack\),
of length \(b_{i}/2 > s': = s/2\); if \(a_{1} \in C(A)\), add
likewise \(S_{0}\) at height \(a_{1} \geq s\). By Jung's theorem
\(A \subseteq \overline{B}\left( c,D/\sqrt{3} \right)\) for the
circumcentre \(c\), and every point of every \(S_{i}\) lies within
\(b_{i} \leq D/\sqrt{3}\) (Lemma~\ref{lem:contact}(iii)) of a point of \(A\); hence
all segments lie in \(B_{R}(c)\) with \(R = 2D/\sqrt{3} \geq s'\).
For \(i < j\), the heights \(h_{i} < h_{j}\) of the associated critical
points differ by at least
\(\varepsilon_{i + 1} + \cdots + \varepsilon_{j}\), since the gaps
\(I_{i + 1},\ldots,I_{j}\) are disjoint subintervals of
\(\left( h_{i},h_{j} \right)\); Lemma~\ref{lem:rsm2}(i) then gives
\[\operatorname{dist}\left( S_{i},S_{j} \right) \geq \sqrt{s/2}\sqrt{\varepsilon_{i + 1} + \cdots + \varepsilon_{j}} \geq \sqrt{s'/2}\sqrt{\varepsilon_{i + 1} + \cdots + \varepsilon_{j}}.\]
Lemma~\ref{lem:rsm2}(ii) yields, when \(S_{0}\) exists,

\[\sum_{i = 1}^{m}\sqrt{\varepsilon_{i}}\, \leq \, 8(s/2)^{- 3/2}\pi\left( 2D/\sqrt{3} \right)^{2}\, = \,\frac{64\sqrt{2}}{3}\,\pi\, s^{- 3/2}D^{2}\, \leq \, 31\pi\, s^{- 3/2}D^{2},\]

and when \(a_{1} \notin C(A)\) (necessarily \(a_{1} = s\)), the same
bound applies to \(\varepsilon_{2},\ldots,\varepsilon_{m}\) while
\(\sqrt{\varepsilon_{1}} \leq \sqrt{D/\sqrt{3}} \leq s^{- 3/2}D^{2}\)
using \(s \leq D\).
\end{proof}

\begin{remark}[relation to {{\cite[Corollary 4.1]{RSM09}}}]\label{rem:rsmcor}
Lemma~\ref{lem:gaps} is a variant of \cite[Corollary 4.1]{RSM09}, which asserts
the bound \(\sum_{i}\sqrt{\varepsilon_{i}} \leq 8s^{- 3/2}\pi D^{2}\) for the
gaps of \((s,r) \smallsetminus C(A)\), for any \(0 < s < r\), and is stated
there without printed proof. We record two points of difference, in the spirit
of making the present paper self-contained rather than of correcting the
literature; the proof above is complete and independent, and nothing below
depends on the value of the constant.

First, on scope. As literally stated for arbitrary \(r\), the decomposition of
\((s,r) \smallsetminus C(A)\) includes the final interval
\(\left( \sup C(A),r \right)\), whose square-root length grows without bound in
\(r\) while the right-hand side does not; since \(\sup C(A) \leq D/\sqrt{3}\) by
Lemma~\ref{lem:contact}(iii), the inequality therefore cannot hold uniformly in
\(r\) unless the gaps are restricted. The meaningful content --- and the way the
corollary is used both in \cite{RSM09} and here --- concerns the gaps with
critical right endpoint, which is how Lemma~\ref{lem:gaps} is stated.

Second, on the constant. Reconstructing the estimate from
Lemma~\ref{lem:rsm2} gives a larger numerical factor than \(8\pi\), for two
reasons that the reader can check against the proof above. The segments
\(S_{i} = \left\lbrack x_{i},\left( x_{i} + p_{i} \right)/2 \right\rbrack\)
attached to critical points at heights \(b_{i} > s\) have length \(b_{i}/2 > s/2\),
not \(s\), so Lemma~\ref{lem:rsm2}(ii) must be applied with \(s' = s/2\),
contributing a factor \(2^{3/2}\); and those segments lie not in a ball of
radius \(D\) but in one of radius \(R = 2D/\sqrt{3}\), by Jung's theorem
together with the height bound \(d_{A} \leq D/\sqrt{3}\) of
Lemma~\ref{lem:contact}(iii), contributing a factor \(4/3\). Hence the
reconstruction yields
\[8\,(s/2)^{- 3/2}\pi\left( 2D/\sqrt{3} \right)^{2}\  = \ \frac{64\sqrt{2}}{3}\,\pi\, s^{- 3/2}D^{2}\  \approx \ 30.2\,\pi\, s^{- 3/2}D^{2},\]
which is the bound proved above, stated as \(32\pi\) for roundness. (Using the
sharper constant \(5\sqrt{2}\) that the proof of \cite[Lemma 4.2]{RSM09}
actually delivers, in place of the \(8\) at which it is stated, gives
\(\left( 80/3 \right)\pi \approx 26.7\,\pi\) instead.)

The exponent \(1/2\) matches Theorem~\ref{thm:fu} exactly: the critical-value
set of a compact planar set is at most half-dimensional, and its gaps are
square-root summable. The general statement is now available in \cite{RZ20}:
by \cite[Theorem 1.1]{RZ20}, for every \(s > 0\) the set
\(C(A) \cap \lbrack s,\infty)\) is a compact null set with finite degree-\(1/2\)
gap sum, and every such subset of \(\lbrack s,\infty)\) is contained in
\(C(A)\) for some compact \(A\) --- optimality in this embedding sense; and
\cite[Proposition 3.16]{RZ20} gives the stronger weighted dyadic bound
\(\sum_{n \geq 0}\delta_{n}^{3/2}\, G_{1/2}\left( C(A) \cap \lbrack\delta_{n + 1},\delta_{n}\rbrack \right) \leq 10^{10}D^{2}\)
with \(\delta_{n} = D2^{- n}\), where \(G_{1/2}\) denotes the sum of the square
roots of the gap lengths. The latter yields a bound of the form of
Lemma~\ref{lem:gaps} with a much larger constant, and also controls the gaps
near \(0\), which Lemma~\ref{lem:gaps} does not.
\end{remark}
\section{The integrated bound}\label{sec:integrated}

\begin{theorem}[Theorem~\ref{thm:A}]\label{thm:main}
 For every compact
\(A \subseteq \mathbb{R}^{2}\) and every \(r_{0} > 0\),

\[\int_{r_{0}}^{\infty}\beta_{1}\left( A^{(r)} \right)dr\, \leq \, 4050\,\frac{D^{4}}{r_{0}^{3}},\]

the integrand vanishing for \(r \geq D/\sqrt{3}\).
\end{theorem}

\begin{proof}
If \(C(A) = \varnothing\) --- as happens for a convex \(A\) --- then
Lemma~\ref{lem:contact}(iii) gives \(\beta_{1}\left( A^{(r)} \right) = 0\) for
every \(r > 0\) and there is nothing to prove; so assume
\(C(A) \neq \varnothing\), whence \(0 < \sup C(A) \leq D/\sqrt{3}\). If
\(D < r_{0}\) then \(r_{0} > D \geq D/\sqrt{3}\), so by
Lemma~\ref{lem:contact}(iii) the integrand vanishes on
\(\left\lbrack r_{0},\infty \right)\); the left-hand side is then \(0\) and the
inequality is trivial, the right-hand side being nonnegative. So
assume \(D \geq r_{0}\) and integrate over
\(\left\lbrack r_{0},D \right\rbrack\). By Theorem~\ref{thm:fu}, almost every
\(r \in \left\lbrack r_{0},D \right\rbrack\) lies in a maximal open gap
of \(\left( r_{0},\infty \right) \smallsetminus C(A)\). On the final gap
\(\left( \sup C(A),\infty \right)\) the integrand vanishes: by
Lemma~\ref{lem:contact}(iii), a bounded complementary component at radius \(r\)
would force a critical value strictly above \(r\), so
\(\beta_{1}\left( A^{(r)} \right) = 0\) for every \(r \geq \sup C(A)\). (It is
this critical-value statement, and not the numerical threshold
\(D/\sqrt{3}\), that is needed here: the final gap may begin well below
\(D/\sqrt{3}\).) Let \(r\) lie in a gap
\(I_{i} = \left( a_{i},b_{i} \right)\) with \(b_{i} \in C(A)\) and set
\(\varepsilon_{r}: = \min\left\{ b_{i} - r,\ r/2 \right\} \in (0,r)\);
then \(\left\lbrack r,r + \varepsilon_{r} \right) \subseteq I_{i}\) is
free of critical values, and Lemma~\ref{lem:rsm1} with \(D + 2r \leq 3D\) gives

\[\beta_{1}\left( A^{(r)} \right)\, \leq \,\frac{9\pi D^{2}}{r_{0}\sqrt{2r_{0}}}\,\varepsilon_{r}^{- 1/2}\, \leq \,\frac{9\pi D^{2}}{r_{0}\sqrt{2r_{0}}}\left\lbrack \left( b_{i} - r \right)^{- 1/2} + \left( r_{0}/2 \right)^{- 1/2} \right\rbrack.\]

Integrating, the second term contributes at most
\(9\pi D^{3}r_{0}^{- 2}\), while
\[\int_{I_{i} \cap \left\lbrack r_{0},D \right\rbrack}^{}\left( b_{i} - r \right)^{- 1/2}dr \leq 2\sqrt{\varepsilon_{i}}\]
and Lemma~\ref{lem:gaps} with \(s = r_{0}\) give for the first term at most
\[\left( 9\pi D^{2}/r_{0}\sqrt{2r_{0}} \right) \cdot 64\pi r_{0}^{- 3/2}D^{2} = \left( 576/\sqrt{2} \right)\pi^{2}D^{4}r_{0}^{- 3}.\]
Since \(D \geq r_{0}\), the total is at most
\(\left( \left( 576/\sqrt{2} \right)\pi^{2} + 9\pi \right)D^{4}r_{0}^{- 3} \leq 4050\, D^{4}r_{0}^{- 3}\).
\end{proof}

\begin{remark}[on sharpness]\label{rem:sharp}
 The constant is not optimized and
the exponents are not sharp: for the \(N \times N\) square grid with
spacing comparable to \(2r_{0}\), the integral is of order
\(N^{2}r_{0} \asymp D^{2}/r_{0}\), while the bound is of order
\(D^{4}/r_{0}^{3}\). We expect the truth to be an area bound,
\(\int_{r_{0}}^{\infty}\beta_{1}\left( A^{(r)} \right)dr \leq C\left| A^{\left( r_{0} \right)} \right|/r_{0}\)
or similar, but the natural packing proof fails on configurations of
radius-ephemeral holes recurring at one location, and we leave the sharp
form open; for the applications in \cite{G26} any polynomial dependence on
\(D\) suffices.
\end{remark}

\subsection{Two further consequences}\label{sec:further}

The critical-value calculus above yields two more facts about planar parallel
sets that are used in \cite{G26} and are of independent interest: the
threshold \(D/\sqrt{3}\) of Lemma~\ref{lem:contact}(iii) can be reached
without critical values, by Jung's theorem, and the boundary length of a
parallel set is controlled by its area.

\begin{lemma}[parallel sets are star-shaped above \(D/\sqrt{3}\)]\label{lem:star}
~ Let
\(K \subset \mathbb{R}^{2}\) be compact nonempty with \(d = \operatorname{diam}K\),
let \(r \geq d/\sqrt{3}\) and let \(c\) be the centre of the smallest
disc containing \(K\). Then \(K^{(r)}\) is star-shaped about every
\(p\) with \(\lVert p - c \rVert \leq r - d/\sqrt{3}\); in particular it is
contractible and \(\beta_{1}K^{(r)} = 0\).
\end{lemma}

\begin{proof}
By Jung's theorem
\(K \subseteq B\left( c,d/\sqrt{3} \right)\). For \(x \in K^{(r)}\)
pick \(q \in K\) with \(\lVert x - q \rVert \leq r\); convexity of
\(u \mapsto \lVert u - q \rVert\) gives, for \(y \in \lbrack p,x\rbrack\),
\(\lVert y - q \rVert \leq \max\{\lVert p - q \rVert,\lVert x - q \rVert\} \leq \max\{\lVert p - c \rVert + d/\sqrt{3},\ r\} = r\),
so \(y \in B(q,r) \subseteq K^{(r)}\). \end{proof}

Lemma~\ref{lem:star} gives a second proof of the last assertion of
Lemma~\ref{lem:contact}(iii), and says more: above \(D/\sqrt{3}\) the
parallel set is not merely hole-free but contractible.

\begin{proposition}[perimeter]\label{prop:perimeter}
For every compact \(A \subset \mathbb{R}^{2}\) and every \(r \notin C(A)\)
at which \(s \mapsto \left| A^{(s)} \right|\) is differentiable,
\[\mathcal{H}^{1}\left( \partial A^{(r)} \right)\  \leq \ \frac{2\left| A^{(r)} \right|}{r}.\]
The exceptional set is the union of \(C(A)\), which is Lebesgue-null by
Theorem~\ref{thm:fu}, and the set of radii at which the nondecreasing function
\(s \mapsto \left| A^{(s)} \right|\) fails to be differentiable, which is
Lebesgue-null as well; so the bound holds at almost every \(r > 0\).
\end{proposition}

\begin{proof}
Kneser's inequality \cite{Kne51}, which for a compact
\(A \subset \mathbb{R}^{d}\) and \(0 < s \leq t\) states that
\(\left| A^{(t)} \right| \leq (t/s)^{d}\left| A^{(s)} \right|\), makes
\(s \mapsto \left| A^{(s)} \right|/s^{2}\) non-increasing in the plane; see
\cite{Sta76} for this and for the regularity of the parallel volume function.
Hence at a point of differentiability
\(\frac{d}{ds}\left| A^{(s)} \right| \leq 2\left| A^{(r)} \right|/r\), and
the derivative equals \(\mathcal{H}^{1}\left( \partial A^{(r)} \right)\)
by \cite[Theorem 3.3]{RSS09}. \end{proof}

\section{A scale-uniform bound through the maxima of the distance function}\label{sec:maxima}

Theorem~\ref{thm:main} controls an integral over \(r\). For statements uniform
in the radius the right deterministic object is a count of components of local
maxima. Let \(\mathcal{M}_{A} \subseteq \mathbb{R}^{2} \smallsetminus A\) be
the set of local maxima of \(d_{A}\); the following elementary lemma gives
each of its components a height.

\begin{lemma}[heights of maxima components]\label{lem:maxima}
Let \(f\) be a continuous real function on a topological space \(X\), and let
\(M\) be the set of local maxima of \(f\). Then \(f\) is constant on every
connected component of \(M\). In particular \(d_{A}\) has a well-defined height
\(h(C)\) on every connected component \(C\) of \(\mathcal{M}_{A}\).
\end{lemma}

\begin{proof}
Let \(C\) be a component of \(M\) and suppose \(f|_{C}\) is not constant.
Choose \(a\) strictly between two values attained on \(C\) and put
\(C_{\leq a}: = \{ x \in C:f(x) \leq a\}\). This set is closed in \(C\) by
continuity. It is also open in \(C\): at a point \(x\) with \(f(x) < a\) by
continuity, and at a point \(x\) with \(f(x) = a\) because \(x\) is a local
maximum, so \(f \leq f(x) = a\) on a neighbourhood of \(x\). Thus
\(C_{\leq a}\) is a nonempty proper subset of \(C\) that is both open and
closed in \(C\), contradicting connectedness.
\end{proof}

For \(s > 0\) let \(m_{A}(s) \in \{ 0,1,\ldots\} \cup \{\infty\}\) be the
number of connected components of \(\mathcal{M}_{A}\) whose (componentwise
constant) height exceeds \(s\):
\[m_{A}(s): = \#\left\{ C\text{ a component of }\mathcal{M}_{A}:\ h(C) > s \right\}.\]

\begin{theorem}[Theorem~\ref{thm:C}]\label{thm:maxima}
 For every compact
\(A \subseteq \mathbb{R}^{2}\) and every \(s > 0\),

\[\sup_{r \geq s}\ \beta_{1}\left( A^{(r)} \right)\, \leq \, m_{A}(s).\]
\end{theorem}

\begin{proof}
Fix \(r \geq s\) and let \(\Omega\) be a bounded connected
component of \(\left\{ d_{A} > r \right\}\); it is open with
\(\partial\Omega \subseteq \left\{ d_{A} = r \right\}\). The maximum of
\(d_{A}\) over the compact \(\overline{\Omega}\) exceeds \(r\), so it is
attained at an interior point \(x^{*} \in \Omega\), which is then a
local maximum of height \(> r \geq s\). Let \(C\) be its component in
\(\mathcal{M}_{A}\); by Lemma~\ref{lem:maxima},
\(d_{A} \equiv d_{A}\left( x^{*} \right) > r\) on \(C\), so
\(C \subseteq \left\{ d_{A} > r \right\}\), and being connected and
meeting \(\Omega\), \(C \subseteq \Omega\). Distinct components
\(\Omega\) are disjoint, so the assignment \(\Omega \mapsto C\) is
injective into the components counted by \(m_{A}(s)\).
\end{proof}

\begin{remark}\label{rem:maxima}
 In the language of persistence, \(m_{A}(s)\)
dominates the number of degree-one classes of the offset filtration
dying above \(s\): each hole disappears, as \(r\) grows, onto a maximum
component. Theorem~\ref{thm:maxima} converts every question about
\(\sup_{r \geq s}\beta_{1}(A^{(r)})\) into a question about the single
quantity \(m_{A}(s)\). For a random compact set, the moments of \(m_{A}(s)\)
are a strictly stronger requirement than those of
\(\int_{s}^{\infty}\beta_{1}(A^{(r)})\,dr\): short-lived holes are invisible
to any countable family of radii but are counted by \(m\). Section~\ref{sec:combs}
shows that \(m_{A}(s)\) cannot be controlled deterministically; its
probabilistic estimation for the Brownian trace is discussed in
\cite[Section 8]{G26}. Lemma~\ref{lem:contact}(ii) --- an empty disc of radius
\(> s\) together with two points of \(A\) on its boundary circle at mutual
distance \(\geq \sqrt{3}\, s\) --- describes the configuration to be counted.
\end{remark}
\section{No pointwise bound: the two-combs family}\label{sec:combs}

Theorem~\ref{thm:main} bounds an integral of the hole count over radii, not
the hole count at a prescribed radius. This section shows that the
restriction is not an artefact of the proof: the strategy of bounding
\(\beta_{1}(A^{(r)})\) at a fixed radius by any of the usual coarse geometric
functionals of \(A\) cannot succeed. We first fix the notion of increment count
that appears in Theorem~\ref{thm:B}, and record why it is a natural candidate.

\subsection{Increments and the decomposition into contractible pieces}\label{sec:pieces}

Let \(\gamma:\lbrack 0,\ell\rbrack \rightarrow \mathbb{R}^{2}\) be a
continuous path and \(\rho > 0\). Its increment times at scale \(\rho\) are
\(\theta_{0} = 0\) and
\(\theta_{j + 1} = \inf\left\{ t > \theta_{j}: \lVert \gamma(t) - \gamma(\theta_{j}) \rVert = \rho \right\}\),
with \(\inf\varnothing = \infty\); \(N_{\rho}(t) = \max\{ j:\theta_{j} \leq t\}\)
is the number of increments completed by time \(t\), and
\(c_{j} = \gamma(\theta_{j - 1})\) the anchor of the \(j\)-th increment. We
write \(N_{\rho}(\gamma) = N_{\rho}(\ell)\) for the increment count of the
whole path, and \(K_{t} = \gamma(\lbrack 0,t\rbrack)\) for its range up to
time \(t\). For a rectifiable curve traversed at unit speed, \(N_{\rho}\)
counts the successive displacements of size \(\rho\) along the curve.

\begin{proposition}\label{prop:pieces}
 Let \(r \geq 2\rho\) and
\(Y_{j}(r): = \bigcup_{s \in \left\lbrack \theta_{j - 1},\,\theta_{j} \land t \right\rbrack}\overline{B}\left( \gamma(s),r \right)\).
Then every ball in this union contains
\(\overline{B}\left( c_{j},r - \rho \right)\); consequently
\(Y_{j}(r)\) is star-shaped with respect to every point of that
disc, hence compact, connected and contractible, with
\({\check{H}}^{1}\left( Y_{j}(r) \right) = 0\); moreover
\(\overline{B}\left( c_{j},r \right) \subseteq Y_{j}(r) \subseteq \overline{B}\left( c_{j},r + \rho \right)\),
the pieces cover:
\(K_{t}^{(r)} = \bigcup_{j = 1}^{N_{\rho}(t) + 1}Y_{j}(r)\), and
consecutive pieces intersect.
\end{proposition}

\begin{proof}
For
\(s \in \left\lbrack \theta_{j - 1},\theta_{j} \land t \right\rbrack\),
\(\lVert \gamma(s) - c_{j} \rVert \leq \rho\); if
\(\lVert y - c_{j} \rVert \leq r - \rho\) then
\(\lVert y - \gamma(s) \rVert \leq r\), proving the common-disc claim.
A union of convex sets with a common point \(y\) is star-shaped with
respect to \(y\), and star-shaped compacta are contractible, so their
\v{C}ech cohomology is that of a point by homotopy invariance. The
inclusions are the ball at \(s = \theta_{j - 1}\) and the triangle
inequality; the covering and chaining are immediate since
\(\gamma(\theta_{j})\) terminates piece \(j\) and initiates piece \(j + 1\).
\end{proof}

Combining Proposition~\ref{prop:pieces} with Lemma~\ref{lem:mv} by induction
along the pieces gives, for every \(r \geq 2\rho\),
\begin{equation}\label{eq:pieces}
\beta_{1}\left( K_{t}^{(r)} \right)\, \leq \,\sum_{m = 2}^{N_{\rho}(t) + 1}\left\lbrack \beta_{0}\left( \left( \bigcup_{j < m}Y_{j}(r) \right) \cap Y_{m}(r) \right) - 1 \right\rbrack .
\end{equation}
One might hope to bound the right-hand side by a function of
\(N_{\rho}(t)\) alone --- for instance quadratically, by counting the pairs of
pieces that meet. The family below shows that no bound of this kind holds: for
the curves \(\gamma_{\varepsilon}\) at \(r = 1\) and \(\rho = 1/2\), the sum
in \eqref{eq:pieces} has at most \(30\) terms and exceeds \(1/(2\varepsilon)\),
so a single term can be arbitrarily large.

\subsection{The two-combs family}\label{sec:twocombs}

Fix \(\rho = 1/2\) and count increments \(N_{\rho}\) as above, for a curve
traversed at unit speed. Theorem~\ref{thm:B} is the following statement
together with Corollary~\ref{cor:combs}.

\begin{theorem}[two combs]\label{thm:combs}
 For every
\(\varepsilon \in (0,1/10\rbrack\) there is a connected
rectifiable curve
\(\gamma_{\varepsilon} \subset \mathbb{R}^{2}\), contained in a
fixed ball of radius \(4\), with
\(\operatorname{length}\left( \gamma_{\varepsilon} \right) \leq 12\)
and \(N_{\rho}\left( \gamma_{\varepsilon} \right) \leq 30\) at
\(\rho = 1/2\), such that

\[\beta_{1}\left( \gamma_{\varepsilon}^{(1)} \right)\, \geq \,\frac{1}{2\varepsilon}.\]
\end{theorem}

\begin{proof}
Set \(H = 0.44\), \(\delta = \varepsilon^{2}/2\),
\(K = \left\lfloor H/\varepsilon \right\rfloor\),
\(y_{k} = k\varepsilon\) for \(|k| \leq K\), and
\(b_{1} = 2 + 2\delta - \varepsilon^{2}/16\). Comb \(A\) is the spine
\(\{ 0\} \times \lbrack - H,H\rbrack\) with teeth
\(\lbrack 0,\delta\rbrack \times \left\{ y_{k} \right\}\); comb \(B\) is
its mirror image with spine
\(\left\{ b_{1} \right\} \times \lbrack - H,H\rbrack\) and teeth
\(\left\lbrack b_{1} - \delta,b_{1} \right\rbrack \times \left\{ y_{k} \right\}\)
at the same heights. The curve \(\gamma_{\varepsilon}\) traces
comb \(A\) from bottom to top (running out and back along each tooth),
follows the detour
\((0,H) \rightarrow (0,4.2) \rightarrow \left( b_{1},4.2 \right) \rightarrow \left( b_{1},H \right)\),
and traces comb \(B\) from top to bottom.

\emph{Book-keeping.} Each vertical leg of the detour has length
\(4.2 - H = 3.76\) and the horizontal leg has length \(b_{1}\), while each comb
pass has length \(2H + (2K + 1) \cdot 2\delta\). Using
\((2K + 1)2\delta \leq \left( 2H/\varepsilon + 1 \right)\varepsilon^{2} = 2H\varepsilon + \varepsilon^{2}\)
and \(b_{1} \leq 2 + \varepsilon^{2}\), the total length is at most
\[2\left( 2H + 2H\varepsilon + \varepsilon^{2} \right) + \left( 3.76 + b_{1} + 3.76 \right)\  \leq \ 2(0.88 + 0.088 + 0.01) + 7.52 + 2.01\  \leq \ 12\]
for \(\varepsilon \leq 1/10\), and the configuration fits in a ball of radius
\(4\). For the increments, at \(\rho = 1/2\): along each comb pass the ordinate
is non-decreasing and the abscissa varies by at most \(\delta\), so consecutive
increment anchors have ordinates differing by at least
\(\sqrt{\rho^{2} - \delta^{2}} \geq 0.49\), giving at most
\(1 + 2H/0.49 \leq 3\) increments per pass; along the three straight detour
segments, at most \(\left\lceil \operatorname{length}/\rho \right\rceil + 1\)
each, that is \(9\), \(6\) and \(9\) respectively. In total
\(N_{\rho} \leq 3 + 3 + 9 + 6 + 9 = 30\).

\emph{The cells.} Work at radius \(1\) in the strip \(|y| \leq H\). For
\(y\) with nearest tooth level \(y_{k}\) and \(u: = y - y_{k}\),
\(|u| \leq \varepsilon/2\), the rightmost point of
\(\gamma_{\varepsilon}^{(1)} \cap \left( \mathbb{R} \times \{ y\} \right)\)
contributed by comb \(A\) is at abscissa
\(f(y) = \delta + \sqrt{1 - u^{2}}\): the tooth \(k\) dominates the
spine because
\(1 - \sqrt{1 - u^{2}} \leq u^{2} \leq \varepsilon^{2}/4 < \delta\), and
dominates farther teeth trivially; the detour contributes nothing at
these heights (see below). Symmetrically comb \(B\) covers abscissae
\(\geq g(y) = b_{1} - \delta - \sqrt{1 - u^{2}}\). Thus the horizontal
line at height \(y\) is entirely covered between the outer boundaries if
and only if \(f(y) \geq g(y)\),
i.e.~\(\sqrt{1 - u^{2}} \geq 1 - \varepsilon^{2}/32\),
i.e.~\(|u| \leq w: = (\varepsilon/4)\sqrt{1 - \varepsilon^{2}/64}\);
call \(\left\{ \left| y - y_{k} \right| \leq w \right\}\) the \(k\)-th
island. Between consecutive islands, \(f(y) < g(y)\) strictly: at
\(u = \varepsilon/2\) the difference is
\(2 - \varepsilon^{2}/16 - 2\sqrt{1 - \varepsilon^{2}/4} \geq 3\varepsilon^{2}/16 > 0\),
using \(\sqrt{1 - x} \leq 1 - x/2\). The open
cell

\[R_{k}\,: = \,\left\{ (x,y):\ y_{k} + w < y < y_{k + 1} - w,\ \ f(y) < x < g(y) \right\},\quad\quad - K \leq k \leq K - 1,\]

is nonempty and disjoint from the unit parallel sets of both combs. It is
also disjoint from that of the detour: every point of \(R_{k}\) has
\(x \geq f(y) \geq \delta + \sqrt{1 - \varepsilon^{2}/4} > 1\), the strict
inequality because
\(1 - \sqrt{1 - \varepsilon^{2}/4} \leq \varepsilon^{2}/4 < \varepsilon^{2}/2 = \delta\),
and, symmetrically, \(b_{1} - x > 1\), while its distance to the
horizontal detour segment at height \(4.2\) is at least \(4.2 - H > 2\);
the nearest points of the two vertical detour segments lie at \((0,H)\)
and \(\left( b_{1},H \right)\), at distance \(\geq x > 1\), respectively
\(\geq b_{1} - x > 1\).

\emph{Enclosure.} Let
\[Q_{k}: = \left\lbrack \frac{1}{2},\ b_{1} - \frac{1}{2} \right\rbrack \times \left\lbrack y_{k},\ y_{k + 1} \right\rbrack.\]
Its boundary lies in \(\gamma_{\varepsilon}^{(1)}\): the vertical edges
are within distance \(\frac{1}{2} < 1\) of the spines; the horizontal
edges are at tooth heights, where
\(f\left( y_{k} \right) = \delta + 1 \geq g\left( y_{k} \right) = b_{1} - \delta - 1\),
so comb \(A\) covers abscissae up to \(\delta + 1\) and comb \(B\) from
\(b_{1} - \delta - 1\) on, jointly the whole edge. Since
\(R_{k} \subseteq \operatorname{int}Q_{k}\), the connected component of
\(\mathbb{R}^{2} \smallsetminus \gamma_{\varepsilon}^{(1)}\) containing
\(R_{k}\) cannot meet \(\partial Q_{k}\) and is therefore a bounded
component contained in \(\operatorname{int}Q_{k}\). The interiors of the \(Q_{k}\) are
pairwise disjoint, so these \(2K\) components are distinct, and
\[\beta_{1}\left( \gamma_{\varepsilon}^{(1)} \right) \geq 2K = 2\left\lfloor H/\varepsilon \right\rfloor \geq 0.88/\varepsilon - 2 \geq 1/(2\varepsilon)\]
for \(\varepsilon \leq 1/10\).

\begin{figure}[htbp]
\centering
\includegraphics[width=0.9\textwidth]{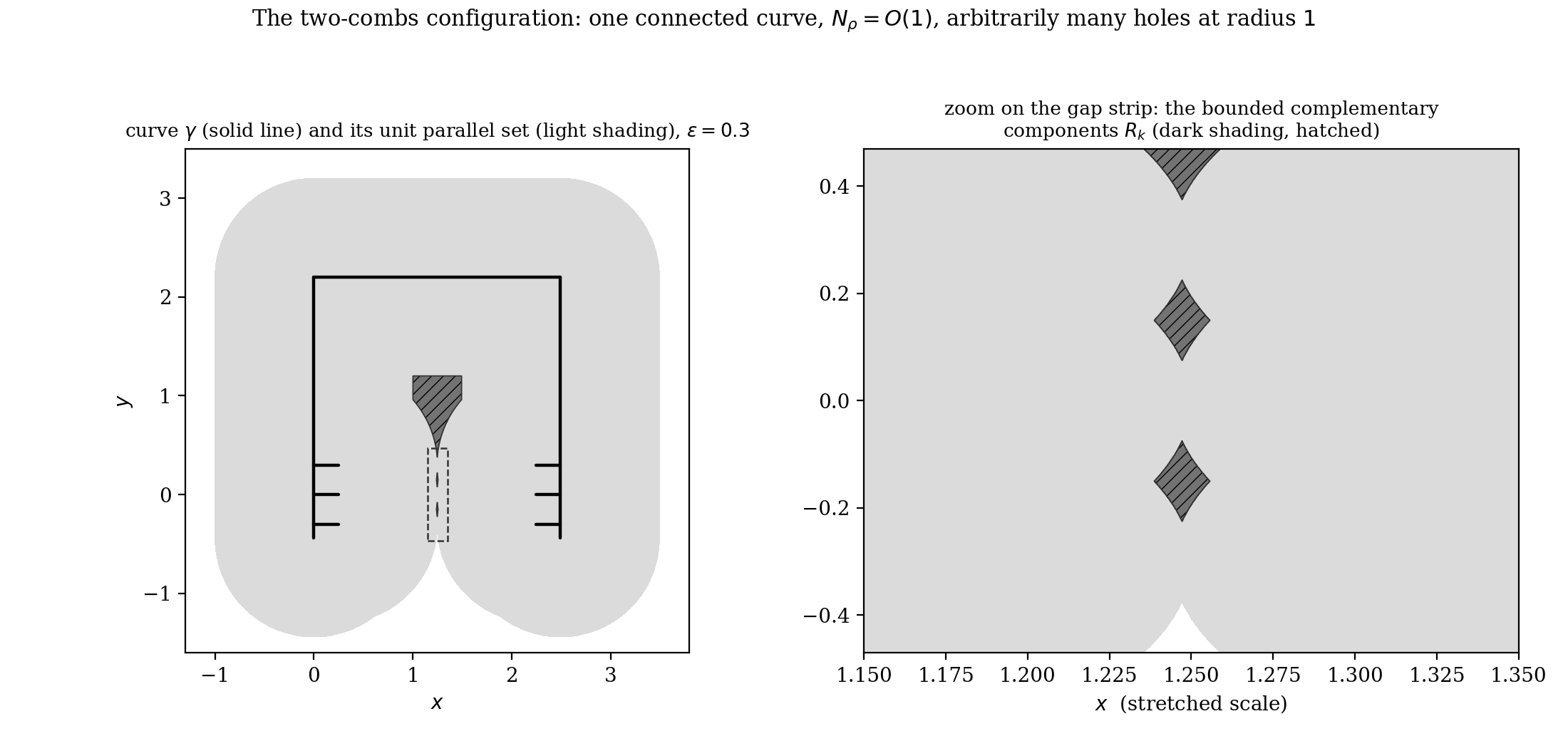}
\caption{Schematic, not to scale. The two-combs configuration, drawn at the
exaggerated value \(\varepsilon = 0.3\) and with a tooth length chosen for
visibility rather than by the relation \(\delta = \varepsilon^{2}/2\) of
Theorem~\ref{thm:combs}, which at this scale would make the cells invisible.
Left: the connected curve (solid line), its radius-\(1\) parallel set (light
shading), the bounded complementary components (dark, hatched) --- the detour
cell at the top and the cells \(R_{k}\) in the strip --- and the zoom box
(dashed). Right: the gap strip at stretched horizontal scale; the hatched
diamonds are the cells \(R_{k}\), and the hatched region at the top edge is
the tip of the detour cell. At small \(\varepsilon\) there are
\(2\lfloor 0.44/\varepsilon\rfloor\) cells.}
\label{fig:two-combs}
\end{figure}

\end{proof}

\begin{corollary}[uniformity of the coarse quantities]\label{cor:combs}
Along the family \(\left( \gamma_{\varepsilon} \right)_{0 < \varepsilon \leq 1/10}\)
of Theorem~\ref{thm:combs},
\[\operatorname{diam}\gamma_{\varepsilon} \leq 8,\quad
\operatorname{length}\left( \gamma_{\varepsilon} \right) \leq 12,\quad
N_{1/2}\left( \gamma_{\varepsilon} \right) \leq 30,\quad
\left| \gamma_{\varepsilon}^{(1)} \right| \leq 25\pi,\quad
\mathcal{H}^{1}\left( \partial\gamma_{\varepsilon}^{(1)} \right) \leq 55 .\]
Consequently, if \(F:\lbrack 0,\infty)^{5} \rightarrow \lbrack 0,\infty)\) is
bounded on bounded sets, the inequality
\[\beta_{1}\left( A^{(r)} \right)\  \leq \ F\left( \operatorname{diam}A,\ \operatorname{length}(A),\ N_{r/2}(A),\ \left| A^{(r)} \right|,\ \mathcal{H}^{1}\left( \partial A^{(r)} \right) \right)\]
fails for some connected rectifiable curve \(A\) with its unit-speed
parametrisation, at every prescribed \(r > 0\).
\end{corollary}

\begin{proof}
The first three bounds are in the proof of Theorem~\ref{thm:combs}. Since
\(\gamma_{\varepsilon}\) lies in a ball of radius \(4\), its unit parallel set
lies in the concentric ball of radius \(5\), which gives the area bound.

For the perimeter, write \(\gamma_{\varepsilon} = P \cup T\), where \(P\) is
the union of the three segments
\(S_{A} = \{ 0\} \times \lbrack - H,4.2\rbrack\),
\(S_{B} = \left\{ b_{1} \right\} \times \lbrack - H,4.2\rbrack\) and
\(\lbrack 0,b_{1}\rbrack \times \{ 4.2\}\) (each spine together with the
adjacent vertical leg of the detour, and the horizontal leg), and \(T\) is the
union of the \(2(2K + 1)\) teeth. Then
\(\gamma_{\varepsilon}^{(1)} = P^{(1)} \cup T^{(1)}\), and
\[\partial\left( P^{(1)} \cup T^{(1)} \right)\  \subseteq \ \partial P^{(1)} \cup \left( \partial T^{(1)} \smallsetminus P^{(1)} \right),\]
since a boundary point of the union that lies in \(P^{(1)}\) is not interior
to \(P^{(1)}\), while one that does not lies in \(T^{(1)}\) and is not
interior to \(T^{(1)}\). The unit parallel set of a segment of length
\(\ell\) is a stadium of perimeter \(2\ell + 2\pi\), and the boundary of a
finite union is contained in the union of the boundaries, so
\(\mathcal{H}^{1}\left( \partial P^{(1)} \right) \leq 2\left( 4.64 + 4.64 + b_{1} \right) + 6\pi \leq 42\).

Consider a tooth \(t = \lbrack 0,\delta\rbrack \times \left\{ y_{k} \right\}\)
of comb \(A\) (comb \(B\) is symmetric). Its unit parallel set \(t^{(1)}\) is a
stadium whose boundary consists of the segments
\(\lbrack 0,\delta\rbrack \times \left\{ y_{k} \pm 1 \right\}\) and the unit
semicircles centred at \(\left( 0,y_{k} \right)\) (left) and
\(\left( \delta,y_{k} \right)\) (right). A point of \(\partial t^{(1)}\) outside
\(P^{(1)}\) is at distance \(> 1\) from \(S_{A}\). Points of the left
semicircle are at distance \(1\) from \(\left( 0,y_{k} \right) \in S_{A}\), and
points of the upper segment are within \(\delta\) of
\(\left( 0,y_{k} + 1 \right) \in S_{A}\), so neither contributes. A point
\(q = \left( \delta + \cos\varphi,\ y_{k} + \sin\varphi \right)\) of the right
semicircle with ordinate in \(\lbrack - H,4.2\rbrack\) has
\(\operatorname{dist}\left( q,S_{A} \right) = \delta + \cos\varphi\), which
exceeds \(1\) only on the arc \(\cos\varphi > 1 - \delta\), of length
\(2\arccos(1 - \delta) \leq \pi\sqrt{2\delta} = \pi\varepsilon\), by
\(1 - \cos\theta \geq 2\theta^{2}/\pi^{2}\). Finally, a point of the right
semicircle with ordinate \(y_{k} + \sin\varphi < - H\) is at distance
\(> 1\) from \(S_{A}\) only if
\(\left( \delta + \cos\varphi \right)^{2} + \left( \sin\varphi + a_{k} \right)^{2} > 1\)
with \(a_{k}: = H + y_{k} \geq 0\) and \(\sin\varphi < - a_{k}\), which forces
\(a_{k}^{2} < 2\delta + \delta^{2}\), i.e.\ \(a_{k} < 1.01\,\varepsilon\); as
the \(a_{k}\) are spaced by \(\varepsilon\), at most two teeth of comb \(A\)
satisfy this, and each contributes at most a quarter circle and its lower
segment, of length \(\leq \pi/2 + \delta\) (a point \((x,y_{k} - 1)\) of a
lower segment is at distance \(> 1\) from \(S_{A}\) only if
\((1 - a_{k})^{2} + x^{2} > 1\), which forces \(a_{k} < \delta^{2}\), a
fortiori \(a_{k} < 1.01\,\varepsilon\)). Altogether
\[\mathcal{H}^{1}\left( \partial T^{(1)} \smallsetminus P^{(1)} \right)\  \leq \ 2\left\lbrack (2K + 1)\pi\varepsilon + 2\left( \pi/2 + \delta \right) \right\rbrack\  \leq \ 2\left\lbrack (2H + \varepsilon)\pi + \pi + \varepsilon^{2} \right\rbrack\  \leq \ 13\]
for \(\varepsilon \leq 1/10\), using \((2K + 1)\varepsilon \leq 2H + \varepsilon\),
and \(\mathcal{H}^{1}\left( \partial\gamma_{\varepsilon}^{(1)} \right) \leq 42 + 13 = 55\).

For the last assertion, at \(r = 1\) all five arguments of \(F\) stay in a
fixed bounded set along the family while
\(\beta_{1}\left( \gamma_{\varepsilon}^{(1)} \right) \geq 1/(2\varepsilon)\)
diverges. For general \(r > 0\) apply the dilation \(x \mapsto rx\): it maps
\(\gamma_{\varepsilon}^{(1)}\) onto \(\left( r\gamma_{\varepsilon} \right)^{(r)}\),
preserves \(\beta_{1}\) and the increment count at the rescaled scale
\(\rho = r/2\), multiplies diameter, length and perimeter by \(r\), and area
by \(r^{2}\).
\end{proof}

\begin{remark}[consequences]\label{rem:combs}
(i) The restriction to \(F\) bounded on bounded sets is needed --- an
arbitrary function need not be --- and is satisfied by every standard
expression in these quantities, in particular by every polynomial or monotone
continuous one. Corollary~\ref{cor:combs} rules out, in particular, quadratic
bounds in the increment count, packing bounds proportional to
\(\left| A^{(r)} \right|/r^{2}\), and bounds of the form
\(\beta_{1}\left( A^{(r)} \right) \leq C\,\mathcal{H}^{1}\left( \partial A^{(r)} \right)/r\):
no fixed-radius bound can pass through area or boundary measure. (ii) The holes
live on a radius window of width \(\varepsilon^{2}/8\), so the family is
consistent with Theorem~\ref{thm:main} --- the integral it produces is in fact
\(O(\varepsilon)\) plus the contribution of the single detour cell --- and with
Lemma~\ref{lem:rsm1}, whose hypothesis fails precisely on that window: the
configuration plants a cluster of critical values of width
\(O\left( \varepsilon^{2} \right)\) just above \(r = 1\). Realizing such a
cluster requires matching two envelopes to second order; that a Brownian trace
does so only with the frequency consistent with finite expectations is exactly
the probabilistic content of the mean-Euler-characteristic theorem of
\cite[Proposition 4.1]{RSM09} and of the fixed-radius theory of
\cite[Section 8]{G26}.
\end{remark}

\section{Notes and open questions}\label{sec:notes}

\emph{The sharp form of Theorem~\ref{thm:A}.} As Remark~\ref{rem:sharp}
records, the \(N \times N\) grid at spacing comparable to \(2r_{0}\) has
\(\int_{r_{0}}^{\infty}\beta_{1}(A^{(r)})\,dr\) of order \(D^{2}/r_{0}\),
against the bound \(D^{4}/r_{0}^{3}\), and we expect the truth to be an area
bound, \(\int_{r_{0}}^{\infty}\beta_{1}(A^{(r)})\,dr \leq C|A^{(r_{0})}|/r_{0}\)
or similar. The natural packing argument --- each hole of \(A^{(r)}\) contains
an empty disc, and disjoint empty discs of radius \(\geq r_{0}\) are few ---
bounds the number of holes that survive a range of radii, but not the number
of radius-ephemeral holes that recur at one location, and the two-combs family
shows that such recurrence is possible; a proof of the area form would have to
count holes with a weight adapted to their lifetime.

\emph{Why the gap route cannot be sharpened.} The proof of
Theorem~\ref{thm:main} loses in two places: in Lemma~\ref{lem:rsm1}, whose
\(\varepsilon^{-1/2}\) singularity may be far from the truth for a particular
set, and in Lemma~\ref{lem:gaps}. The exponent \(1/2\) in the second step
cannot be uniformly reduced: by \cite[Theorem 1.1]{RZ20} every compact null
subset of \(\lbrack s,\infty)\) with summable square-root gaps is contained in
\(C(A)\) for some compact \(A\), so no summability of a smaller power of the
gaps holds in general, and
any improvement of Theorem~\ref{thm:main} must come from the way holes are
counted against critical values, not from the size of \(C(A)\).

\emph{Second moments.} The proof of Theorem~\ref{thm:main} bounds
\(\beta_{1}(A^{(r)})\) by a multiple of \(\varepsilon_{r}^{-1/2}\), where
\(\varepsilon_{r}\) is the distance from \(r\) to the next critical value, and
Lemma~\ref{lem:gaps} makes \(\varepsilon_{r}^{-1/2}\) integrable. The same
route cannot bound \(\int_{r_{0}}^{\infty}\beta_{1}(A^{(r)})^{2}\,dr\), since
\(\varepsilon_{r}^{-1}\) is not integrable across any gap: the calculus caps
the usable negative moment of the gap at the exponent \(1/2\) of Fu's theorem.
Whether a deterministic bound on the integrated square exists we do not know.
The two-combs family does not decide it: there the holes live on a radius
window of width \(\varepsilon^{2}/8\) and number about \(0.88/\varepsilon\),
so the integrated square stays bounded along the family, of order \(1/10\).
In the Brownian setting the corresponding question --- whether the hole count
of the Wiener sausage at a fixed radius has a finite second moment --- is
open; see \cite[Section 8]{G26}.

\emph{Dimension.} The argument is planar throughout: it rests on Alexander
duality in \(S^{2}\), which identifies \(\beta_{1}\) with a count of
complementary components and makes \(\check{H}^{2}\) vanish, and on the
half-dimensionality of the critical-value set, which is sharp in the
dimensional sense \cite[Section 5.2]{Fu85}. Ferry \cite{Fer76} exhibits a
compact set in \(\mathbb{R}^{4}\) whose parallel sets fail to have manifold
boundaries at noncritical radii, so already the bridge lemma
(Lemma~\ref{lem:bridge}) does not extend as stated. An integrated bound for
the Betti numbers of parallel sets in \(\mathbb{R}^{d}\) would need
replacements for both ingredients.

\emph{Related work.} The critical-point notion for distance functions goes back
to Grove \cite{Gr93} in the Riemannian setting and to Fu \cite{Fu85} for
compact subsets of Euclidean space; the smallest positive critical value, the
weak feature size of Chazal and Lieutier \cite{CL07}, and the sampling theory
of \cite{CCL09} exploit the constancy of the topology of parallel sets on
critical-value-free intervals, which is the same mechanism that Lemma~\ref{lem:rsm1}
quantifies. The smallness of \(C(A)\) for planar compacta is studied in
depth by Rataj and Zaj\'{\i}\v{c}ek \cite{RZ20}, who show that the sets
\(C(A) \cap \lbrack s,\infty)\) are exactly as small as the compact null sets
with summable square-root gaps, in the embedding sense above, determine the
Minkowski dimension of
\(C(A)\) as \(4/5\) in the worst case, and prove
\(\int_{0}^{\infty}\mathcal{H}^{1}\left\{ x\ \text{critical}:d_{A}(x) > r \right\}dr < \infty\)
\cite[Proposition 3.5]{RZ20}, the nearest relative of Theorem~\ref{thm:main}
we know of: an integral over the radius of a measure of the critical set
above height \(r\), where Theorem~\ref{thm:main} integrates a count of holes.
Theorem~\ref{thm:fu} is used in \cite{RSM09} to prove that the planar Wiener
sausage has finitely many holes in expectation at a fixed radius;
Theorem~\ref{thm:main} is the deterministic, integrated counterpart of that
argument, and the Brownian applications are developed in \cite{G26}.

\section*{Declaration of competing interest}

The author declares that there are no known competing financial interests or
personal relationships that could have appeared to influence the work reported
in this paper.

\section*{Declaration of generative AI and AI-assisted technologies in the writing process}

During the preparation of this work the author used Anthropic's Claude in order
to typeset the manuscript source, to check internal cross-references and
bibliographic details, and to assist with editorial revision and with the
drafting of expository text. The results reported here and the mathematical
program they belong to are the author's own. The author has reviewed, checked
and edited all content, and takes full responsibility for the content of the
publication.


\begin{thebibliography}{99}

\bibitem{CCL09} Chazal, F., Cohen-Steiner, D., and Lieutier, A. (2009). A sampling theory for compact sets in Euclidean space. \emph{Discrete \& Computational Geometry} \textbf{41}(3), 461--479.

\bibitem{CL07} Chazal, F., and Lieutier, A. (2007). Stability and computation of topological invariants of solids in \(\mathbb{R}^{n}\). \emph{Discrete \& Computational Geometry} \textbf{37}(4), 601--617.

\bibitem{ref12} Edelsbrunner, H., and Harer, J. (2010). \emph{Computational Topology: An Introduction}. American Mathematical Society, Providence, RI.

\bibitem{Fer76} Ferry, S. (1975/76). When \(\varepsilon\)-boundaries are manifolds. \emph{Fundamenta Mathematicae} \textbf{90}(3), 199--210.

\bibitem{Fu85} Fu, J. H. G. (1985). Tubular neighborhoods in Euclidean spaces. \emph{Duke Mathematical Journal} \textbf{52}(4), 1025--1046.

\bibitem{Gr93} Grove, K. (1993). Critical point theory for distance functions. In \emph{Differential Geometry: Riemannian Geometry} (Los Angeles, CA, 1990), Proceedings of Symposia in Pure Mathematics \textbf{54}, Part 3, 357--385. American Mathematical Society, Providence, RI.

\bibitem{G26} Guillaume, T. (2026). Persistence of the Wiener sausage: sampling stability and limit theorems for Betti-curve functionals of drifted planar Brownian motion. Preprint.

\bibitem{HNPS19} Holden, N., Nacu, \c{S}., Peres, Y., and Salisbury, T. S. (2019). How round are the complementary components of planar Brownian motion? \emph{Annales de l'Institut Henri Poincar\'e, Probabilit\'es et Statistiques} \textbf{55}(2), 882--908.

\bibitem{Ho14} Honzl, O. (2014). On an upper bound of the Euler characteristic of the Wiener sausage. \emph{Methodology and Computing in Applied Probability} \textbf{16}(2), 331--353.

\bibitem{Kne51} Kneser, M. (1951). \"Uber den Rand von Parallelk\"orpern. \emph{Mathematische Nachrichten} \textbf{5}, 241--251.

\bibitem{LG90} Le Gall, J.-F. (1990). Wiener sausage and self-intersection local times. \emph{Journal of Functional Analysis} \textbf{88}(2), 299--341.

\bibitem{Mo89} Mountford, T. S. (1989). On the asymptotic number of small components created by planar Brownian motion. \emph{Stochastics and Stochastics Reports} \textbf{28}(3), 177--188.

\bibitem{RSS09} Rataj, J., Schmidt, V., and Spodarev, E. (2009). On the expected surface area of the Wiener sausage. \emph{Mathematische Nachrichten} \textbf{282}(4), 591--603.

\bibitem{RZ20} Rataj, J., and Zaj\'{\i}\v{c}ek, L. (2020). Smallness of the set of critical values of distance functions in two-dimensional Euclidean and Riemannian spaces. \emph{Mathematika} \textbf{66}, 297--324.

\bibitem{RSM09} Rataj, J., Spodarev, E., and Meschenmoser, D. (2009). Approximations of the Wiener sausage and its curvature measures. \emph{Annals of Applied Probability} \textbf{19}(5), 1840--1859.

\bibitem{Sp81} Spanier, E. H. (1981). \emph{Algebraic Topology}. Springer, New York.

\bibitem{Sta76} Stach\'o, L. L. (1976). On the volume function of parallel sets. \emph{Acta Scientiarum Mathematicarum (Szeged)} \textbf{38}, 365--374.

\end{thebibliography}
\end{document}